\documentclass{sig-alternate-05-2015}

\usepackage[pass]{geometry}
\mathcode`l="8000
\begingroup
\makeatletter
\lccode`\~=`\l
\DeclareMathSymbol{\lsb@l}{\mathalpha}{letters}{`l}
\lowercase{\gdef~{\ifnum\the\mathgroup=\m@ne \ell \else \lsb@l \fi}}%
\endgroup

\newfont{\authfntsmall}{phvr at 11pt}
\newfont{\eaddfntsmall}{phvr at 9pt}

\usepackage[utf8]{inputenc}   % Encodage UTF8, important !
\usepackage[T1]{fontenc}      % Inutile avec le paquet suivant ?
\usepackage[shortlabels]{enumitem}
\setlist[description]{font=\normalfont\itshape,itemsep=0ex,partopsep=0ex}

\usepackage{microtype}  % Belle typo
\usepackage{flushend}   % Équilibre les colonnes de la dernière page
\usepackage{mathtools}  % Pour mathrlap et compagnie
\usepackage{booktabs}   % De beaux tableaux
\usepackage{array}      % Pour définir certains types de colonnes
\usepackage[nocompress,nosort]{cite}

\usepackage{float}
\usepackage[noend]{algpseudocode}
\newfloat{algo}{tb}{lop}
\floatname{algo}{Algorithm}
\usepackage{caption}  % La légende des boites flottantes
\usepackage{comment}  % Commentaires multilignes
\usepackage{amssymb,amsmath}
\usepackage{amsfonts}
\usepackage{xcolor}
\usepackage{url}
\usepackage[plainpages=false,pdfpagelabels,colorlinks=true,citecolor=blue,hypertexnames=false]{hyperref}
\usepackage{paralist}

\excludecomment{longversion}

\usepackage{theorem}

\newtheorem{thm}{Theorem}
\newtheorem{prop}[thm]{Proposition}
\newtheorem{lem}[thm]{Lemma}

\newdef{example}{Example}

\def\paragraph#1{\smallskip\noindent{\bf #1.}}

\newenvironment{algoenv}[3][\linewidth]{
\begin{minipage}{#1} %
\flushleft
\rule{\textwidth}{.08em}\vspace{-\baselineskip}\smallskip
\begin{description}[noitemsep]
\item[\rlap{Input}\phantom{Output}] #2
\item[Output] #3
\end{description}
\vspace{-\baselineskip}
\rule{\textwidth}{.05em}
\begin{algorithmic}}
{\end{algorithmic}
\vspace{-.5\baselineskip}
\rule{\textwidth}{.08em}
\end{minipage}}

\newcommand{\lc}{\operatorname{lc}}

\newcommand{\bigO}{{{O}}}
\newcommand{\softO}{\tilde{\bigO}}
\def\le{\leqslant}
\def\ge{\geqslant}
\newcommand{\ndiv}{\hspace{-4pt}\not|\hspace{2pt}}
\newcommand{\cmid}{\hspace{-1pt}\mid\hspace{-1pt}}
\newcommand{\Rdeg}{\operatorname{Rdeg}_n}

\def\KK{\mathbb{K}}
\def\kk{\mathbf{k}}

\def\QQ{\mathbb{Q}}
\def\ZZ{\mathbb{Z}}

\def\gathen#1{{#1}}

\title{Efficient Algorithms for Mixed Creative Telescoping}	

\numberofauthors{3}
\author{
  \alignauthor
  Alin Bostan\\
  \affaddr{Inria}\\
  \affaddr{France}\\
  \email{Alin.Bostan@inria.fr}
  \alignauthor
  Louis Dumont\\
  \affaddr{Inria}\\
  \affaddr{France}\\
  \email{Louis.Dumont@inria.fr}
  \alignauthor
  Bruno Salvy\\
  \affaddr{Inria, Laboratoire LIP}\\
  \affaddr{(U.~Lyon,~CNRS,~ENS~Lyon,~UCBL)}\\
  \affaddr{France}\\
  \email{Bruno.Salvy@inria.fr}
}

\doi{http://dx.doi.org/10.1145/2930889.2930907}

\begin{document}
%\newfont{\mycrnotice}{ptmr8t at 7pt}
%\newfont{\myconfname}{ptmri8t at 7pt}
%\let\crnotice\mycrnotice%
%\let\confname\myconfname%

%\clubpenalty=10000 
%\widowpenalty = 10000

%\setlength{\belowdisplayskip}{.3\baselineskip} \setlength{\belowdisplayshortskip}{0pt}
%\setlength{\abovedisplayskip}{.36\baselineskip} \setlength{\abovedisplayshortskip}{0pt}

\maketitle

\begin{abstract}
Creative telescoping is a powerful computer algebra paradigm --initiated by Doron Zeilberger in the 90's-- for dealing with definite integrals and sums with parameters.  
We address the mixed continuous–discrete case, and focus on the  
integration of bivariate hypergeometric-hyperexponential terms.
We design a new creative telescoping algorithm operating on this class of inputs, based on a Hermite-like reduction procedure. The new algorithm has two nice features: it is efficient and it delivers, for a suitable representation of the input, a minimal-order telescoper.
Its analysis reveals tight bounds on the sizes of the telescoper it produces.
\end{abstract}


% bounds, complexity

\section{Introduction}\label{sec:intro}
%\todo{Problem, existing work, main results}

\paragraph{Context}
\emph{Creative telescoping} is an algorithmic approach introduced in
computer algebra by Zeilberger~\cite{Zeilberger90,Zeilberger91,WiZe92} to
address definite summation and integration for a large class of
functions and sequences involving parameters. 

In this article, we focus on the mixed continuous–discrete case.
Given a term $F_n(x)$ that is both hypergeometric (i.e., $F_{n+1}(x)/F_n(x)$ is a rational function) and hyperexponential (i.e., $F_n'(x)/F_n(x)$ is a rational function), the question is to find a linear recurrence relation satisfied by the sequence of integrals $I_n = \int_\gamma F_n(x) dx$ over a domain~$\gamma$ where~$F_n(x)$ is integrable.
To do this, the method of creative telescoping looks for polynomials $c_0(n), \ldots, c_r(n)$, not all zero, and for a rational function $Q(n,x)$ such that $G_n(x) = Q(n,x) F_n(x)$ satisfies the telescoping relation
\begin{equation}\label{eqn:prob-ct}
 L(F_n(x))\stackrel{\text{def}}{=} \sum_{i=0}^r c_i(n) F_{n+i}(x) =  G_n'(x).
\end{equation}
The recurrence operator~$L=\sum_{i=0}^r c_i(n) S_n^i$ in the shift operator $S_n$ is called a \emph{telescoper} for~$F_n(x)$, and the rational function $Q(n,x)$ is called a \emph{certificate} for the telescoper~$L$.
The integer~$r$ is the \emph{order} of~$L$ and $\max_j \deg c_j$ is its \emph{degree}. 

Throughout the article, the ground field, denoted by~$\kk$, is assumed to be of
characteristic zero. Under suitable additional assumptions, $L(I_n)=0$ is a recurrence relation satisfied by the sequence of
integrals $I_n = \int_\gamma F_n(x) dx$. 

\noindent A subtle point is that not all hypergeometric-hyperexponential terms
admit telescopers. A criterion for deciding when this is the case has been given only recently~\cite[Section~6]{ChChFeFuLi15}: a hypergeometric-hyperexponential term is telescopable if and only if it can be written as the sum of a derivative and of a \emph{proper term}. In our context, proper terms are of the form 
\begin{equation}\label{eqn:proper}
P(n,x) \cdot H(x) ^{\!n} \cdot\, \exp\!\left(\int\!\frac{S(x)}{T(x)}\right)\cdot \Upsilon(n),
\end{equation}
for $P\in\kk[n,x]$, $H\in\kk(x)$, $S,T\in\kk[x]$ and $\Upsilon$ hypergeometric.

Several methods are known for computing a telescoper~$L$ and the corresponding
certificate~$Q$, even in much greater generality~\cite{AlZe90,
WiZe92,Chyzak00,Tefera02,ApZe06,Koutschan13}. Despite a very rich research
activity during the last decade, not much is known about the complexity of
creative telescoping methods when applied to an input of the 
type~\eqref{eqn:proper}. In particular, few estimates are available in the
literature on the order and degree of the minimal-order telescoper. 

If~$\Upsilon(n)\neq 1$, one can compute a telescoper of the form~\eqref{eqn:prob-ct} for the rest and then multiply the coefficient $c_i(n)$ by the rational function $\Upsilon(n)/\Upsilon(n+i)$; this process does not affect the minimality of the telescoper. Thus we focus on the case % when $\Upsilon(n)=1$:
\begin{equation}\label{eqn:proper-bis}
F_n(x)=P(n,x) \cdot H(x)^{\!n} \cdot\, \exp\!\left(\int\!\frac{S(x)}{T(x)}\right).
\end{equation}

%The study of this case is motivated both by its fundamental nature and by its applications~\cite{xxx}.

\paragraph{Previous work}
%$\bullet$ {\bf Creative telescoping.}
%Founded by Zeilberger~\cite{Zeilberger90,Zeilberger91,WiZe92}; 
%extended by Chyzak~\cite{Chyzak00} and Koutschan~\cite{Koutschan13}.
%$\bullet$ {\bf Mixed setting.}
%Almkvist and Zeilberger~\cite[\S8]{AlZe90}; Wilf and Zeilberger~\cite{WiZe92}; Tefera~\cite{Tefera02}; Apagodu and Zeilberger~\cite{ApZe06};
% $\bullet$ {\bf Existence criterion.} 
% Chyzak et alii~\cite{ChChFeFuLi15} provides an algorithmic criterion to determine whether Zeilberger’s algorithm terminates.
% $\bullet$ {\bf Size bounds and complexity.}
Among the various classes of creative telescoping methods, three different
approaches allow to obtain bounds on the sizes of the telescoper and the
certificate, and to control the algorithmic complexity. The first approach is
based on elimination techniques~\cite{Zeilberger90,WiZe92,Yen93}; 
it generally yields
pessimistic bounds, not very efficient algorithms and telescopers of
non-minimal order. The second approach, initiated by Apagodu and
Zeilberger~\cite{ApZe06}, provides sharp bounds on the order and degree of
telescopers, more efficient algorithms, but does not provide information on telescopers of minimal order and works under a restrictive genericity
assumption.
It was successfully applied by Kauers and
co-authors~\cite{ChKa12,ChKa12b,KaYe15} to (bivariate) hyperexponential terms
and hypergeometric terms. The third approach, based on Hermite-like reduction,
is the only one that computes telescopers of minimal order, while guaranteeing
a good control on sizes and complexity. It has been introduced by Bostan \emph{et alii}
for the integration of bivariate rational functions~\cite{BoChChLi10}, then
extended to bivariate hyperexponential functions~\cite{BoChChLiXin13},
multivariate rational functions~\cite{BoLaSa13, Lairez2014} and recently
adapted to summation~\cite{ChHuKa15} for bivariate hypergeometric terms. The
present work is part of the on-going effort in this direction.

\paragraph{Contributions}
We present the first Hermite-style algorithm in the mixed
(continuous–discrete) setting. Our approach is inspired by the proof of
Manivel's lemma~\cite[\S2.3]{Furter15}, originally designed 
in connection with the so-called \emph{polynomial rigidity conjecture}.

Our algorithm works on terms of the form~\eqref{eqn:proper-bis}. 
Its input is $F_n(x)=P(n,x)\Phi(n,x)$, given by
\begin{equation}\label{eqn:input}\tag{{$\star$}}
P(n,x),\quad 
\frac{\Phi'}{\Phi}(n,x)
=\frac{A(n,x)}{B(x)}
=n\frac{H'(x)}{H(x)}+\frac{S(x)}{T(x)},
\end{equation}
i.e., a polynomial in $\kk[n,x]$ and a rational function in $\kk(n,x)$ of a special form given either in lowest terms ($A/B$ with $\gcd(A,B)=1$ and $\deg_nB=0$) or decomposed as the sum of the logarithmic derivative of a rational function~$H\in\kk(x)$ % ($F/G$ in Eq.~\eqref{eqn:proper-bis}) 
multiplied by $n$ and another rational function $S/T\in\kk(x)$. Note that, for a given term, several choices are available for $P$ and $\Phi$. If we further assume that $\Phi'/\Phi$ has no positive integer residue, they become unique and we call them the \emph{minimal decomposition} of $F_n(x)$.

Our main results consist in bounds on the order and degree of telescopers for $F_n(x)$, that are summarized in Theorem~\ref{thm:main theorem} below (see Theorems~\ref{thm:mixedct-order} and \ref{thm:bound-degree} for more precise statements). A complexity analysis of the algorithms leading to these bounds is conducted in section~\ref{sec:complexity} (see in particular Theorem~\ref{thm:global-complexity}).

\smallskip \noindent In the statement of the following theorem, $\deg_xH$ and $\deg_n P$ denote the maximum of the degrees of their numerator and denominator, and $d_H$ is the degree of $H$ at infinity. %, see notation below.

\begin{thm}\label{thm:main theorem}
Given a decomposition as in~\eqref{eqn:input}, $F_n(x)$ 
admits a telescoper of order $r$ bounded by $\delta$, where
\begin{equation}\label{eq:eqdelta}
\delta:=\max(\deg_xA,\deg_x B-1),
\end{equation} 
and degree bounded by
\[r(\deg_xH+\deg_n P)+\deg_x P.\]

Moreover, if the decomposition is minimal then these bounds apply to a minimal-order telescoper for $F_n(x)$.
\smallskip

Algorithm \emph{\textsf{MixedCT}} (\S\ref{sec:mixedCT}) produces a telescoper with these properties. 
If $d$ denotes an upper bound on the degrees of the numerator and denominator of $H$, and of all the polynomials in~\eqref{eqn:input}, and if all these polynomials are square-free, then the telescoper has arithmetic size $\bigO(d^3)$ and \emph{\textsf{MixedCT}} computes it using at most
\begin{equation*}
\begin{cases}
\softO(d^{\omega+1})\quad &\mathrm{if}\ d_H<0,\\
\softO(d^5)\quad &\mathrm{if}\ d_H\ge 0
\end{cases}
\end{equation*}
arithmetic operations in $\kk$, where $\omega$ denotes a feasible matrix multiplication exponent for~$\kk$.
\end{thm}

\noindent Note that one can always choose the decomposition of $F_n(x)$ in $\eqref{eqn:input}$ to be minimal. Indeed, we may write $\Phi=Q\tilde{\Phi}$ where $Q$ is a polynomial and $\tilde{\Phi}'/\tilde{\Phi}$ has no positive integer residue. We then have a minimal decomposition $F_n(x)=(PQ)\tilde{\Phi}$.

The proof of Theorem~\ref{thm:main theorem} is achieved through two main ingredients: \emph{confinement} and \emph{Hermite reduction}. Confinement is the property that given $\Phi$, any polynomial $P$ can be reduced modulo derivatives to a polynomial $R$ of degree at most $\delta-1$: $P\Phi=R\Phi+\Gamma'$ for some~$\Gamma$. It gives a finite dimensional vector space over $\kk(n)$, where the computation will be confined. Hermite reduction is more classical. In this context, it consists in an algorithm that performs a reduction~$P\Phi(n+1,x)=R\Phi(n,x)+\Gamma'$ for some~$\Gamma$. By iterated application of both these operations, all $P(n+i,x)\Phi(n+i,x)$ for $i=0,\dots,\delta$ can be rewritten in a vector space of dimension $\delta$ over $\kk(n)$. Thus by (polynomial) linear algebra there exists a non-trivial linear relation between them, i.e., a telescoper. A more careful study of the increase of the degrees in~$n$ during these rewritings gives the degree bound.

\paragraph{Notation}
In all that follows, $\kk$ and $\KK$ will denote fields of characteristic~0. In our applications we will often set $\KK=\kk(n)$.
We denote by $\KK[x]_d$ the set of polynomials in $\KK[x]$ of degree less than~$d$.

Rational functions are always written in reduced form, with monic denominator. Thus the numerator and denominator are defined without ambiguity. If $k$ is the degree of the numerator and $\ell$ the degree of the denominator of a rational function~$F$, we say that $F$ has \emph{rational degree} $(k,\ell)$, that we denote $\operatorname{Rdeg}(F)=[k]/[\ell]$. We also define the \emph{regular degree} $\deg(F)$ of $F$ as $\deg(F)=\max(k,l)$. Finally, the \emph{degree at infinity} of $F$ is defined as $\deg^\infty(F)=k-l$. The variable with respect to which the degree is taken will be indicated as a subscript when there is an ambiguity.

A polynomial is called \emph{square-free} when its gcd with its derivative is trivial.
The \emph{square-free decomposition} of a monic polynomial $Q\in\KK[x]$ is a factorization~$Q=Q_1^1\dotsm Q_m^m$, with $Q_i\in\KK[x]$ monic and square-free, the $Q_i$'s pairwise coprime and~$\deg_x(Q_m)>0$. The \emph{square-free part} of~$Q$ is the polynomial~$Q^\star=Q_1\dotsm Q_m$.

\paragraph{Structure of the article}
Section~\ref{sec:algos} gives the main properties of the confinement
and of the mixed Hermite-like reduction, leading to the bound on the order of the telescoper.
Section~\ref{sec:degreebounds} gives a bound on the degree of the telescoper and analyzes the evolution of the degrees during the reductions, preparing the complexity analyses in Section~\ref{sec:complexity}.
We conclude the article in Section~\ref{sec:applications} with a few applications of our implementation and experiments on the actual growth of the minimal-order telescopers.

\section{Algorithms and Order Bound}\label{sec:algos}
In this section, we introduce the algorithms with just enough information to prove their correctness and to obtain a bound on the order of the telescoper they compute. A more thorough analysis of the degrees is in the next section.
\subsection{Confinement}\label{sec:confinement}
In terms of integrals, the operation of confinement writes
\[\int_\gamma{P(n,x)\Phi(n,x)\,dx}=\int_\gamma{R(n,x)\Phi(n,x)\,dx},\]
with $R$ a polynomial of degree smaller than $\delta$ from Eq~\eqref{eq:eqdelta}.
This transformation is based on the following lemma.
\begin{lem}\label{lemma:confinement}
Let $\Phi$, $A$ and $B$ be as in Eq.~\eqref{eqn:input}, and $\delta$ as in Eq.~\eqref{eq:eqdelta}. Then for any polynomial $P\in\KK[x]$, there exist unique polynomials $R$ and $Q$ in $\KK[x]$ with $\deg R\le\delta-1$ and $\deg Q\le \deg P-\delta$ such that
\begin{equation}\label{eq:confinement}
P\Phi=R\Phi+(QB\Phi)'.
\end{equation}
\end{lem}
Thus all polynomial multiples of $\Phi$ can be written modulo derivatives on a vector space of dimension $\delta$.
\begin{proof}
Equation~\eqref{eq:confinement} rewrites
\begin{equation}\label{eq:confinement2}
QA + (QB)'=P-R.
\end{equation}
Denote $d=\deg(P)-\delta$, and consider the linear map
\[f\ :\ \KK[x]_{d+1}\rightarrow \KK[x]_{d+1}\quad Q\mapsto (QA+(QB)') \operatorname{div}x^\delta,\]
where $u \operatorname{div} v$ denotes the quotient in the Euclidean division of $u$ by $v$. 
For $m\le d$, $f(x^{m})$ has degree at most $m$.
Its coefficient of degree~$m$ equals either
$\lc A\neq0$ if $\delta=\deg A>\deg B-1$, or
$m+\delta+1\neq0$ if $\delta=\deg B-1>\deg A$, 
or $\lc A+m+\delta+1$ if $\deg A=\deg B-1$. For this case to occur, it is necessary that $\deg S<\deg T$. In that case, write $H'/H=U/V$ with $U,V$ polynomials such that $\deg U=\deg V-1$. Then $A/B=(nUT+VS)/(VT)$, and
\[\deg (VS)\le \deg V+\deg T-1=\deg (UT).\]
From this we see that $\lc(A)$ depends on $n$ and we deduce that $f(x^m)$ has degree exactly $m$ in all cases.

Thus $f$ is an isomorphism. It follows that Equation~\eqref{eq:confinement2} is equivalent to $Q$ being the unique polynomial such that $f(Q)=P \operatorname{div}x^\delta$, and then $R$ is $P-QA-(QB)'$. The degrees of $Q$ and $R$ directly follow from the construction.
\end{proof}

\smallskip\par\noindent \text{Algorithm \textsf{Confinement}} implements the proof of Lemma~\ref{lemma:confinement}. Its correctness follows from the fact that the relation used in the loop is obtained by extracting the coefficient of $x^{i+\delta}$ in Eq.~\eqref{eq:confinement2}.

The key to the minimality in Theorem~\ref{thm:main theorem} is the following property of the confinement.

\begin{prop}\label{lem:confinement soundness property}
	 With the notation of \eqref{eqn:input}, further assume that $\Phi'/\Phi$ has no positive integer residue.
	Then, for any polynomial $R\in\KK[x]$ such that $\deg_xR<\delta$ 
	\[\exists K\in\KK(x)\ R\Phi = (K\Phi)' \Leftrightarrow R = 0.\]
\end{prop}

\begin{proof}
	Only the direct implication is not obvious. If such a $K$ exists, the equation rewrites
	\[R=K'+K\frac{A}{B}.\]
	If $K$ is a polynomial, this equality can only be satisfied if $B$ divides $K$, in which case the result is a direct consequence of the uniqueness in Lemma~\ref{lemma:confinement}.
	Now assume that $K$ has a pole $x_0$ of order $v > 0$. Then the equation
	shows that $A/B$ must have a simple pole at $x_0$ with residue $v$, which contradicts the assumption on the residues of $\Phi'/\Phi$.
\end{proof}

\begin{algo}
	Algorithm \textsf{Confinement}($P$,$F$) 
	
	\begin{algoenv}{A polynomial $P\in\KK[x]$, a rational function $F=A/B$ with $\gcd(A,B)=1$.}{A polynomial $R\in\KK[x]$ of degree less than $\max(\deg(A),\deg(B)-1)$ such that $P=R+(QB)'+QA$ for some $Q\in\KK[x]$.}
		\State $\delta\gets \max(\deg(A),\deg(B)-1)$;
		\State $d\gets \deg(P)-\delta$;
		\State Write $A=\sum_{i}{a_ix^i}$, $B=\sum_{i}{b_ix^i}$, $P=\sum_{i}{p_ix^i}$;
		\For {$i\gets d$ to $0$ by $-1$}
			\State $c\gets a_\delta + (\delta+i+1)b_{\delta+1}$;
			\State $q_i \gets \frac{1}{c}\left(p_{\delta+i} - \sum_{j=1}^{\delta}{q_{i+j}a_{\delta-j}} \right.$
        \State\qquad\qquad $\left.- (\delta+i+1)\sum_{j=1}^{\delta+1}{q_{i+j}b_{\delta+1-j}}\right)$;
		\EndFor
		\State $Q \gets \sum_{i=0}^d{q_ix^i}$;
		\State \Return $P-(QB)'-QA$.
		
	\end{algoenv}
%	\caption{Confinement\label{algo:confinement}}
\end{algo}

\subsection{Hermite Reduction}\label{sec:Hermite}

\begin{algo}
Algorithm \textsf{BasicReduction}($P$,$F$,$G$,$k$)

\begin{algoenv}{A polynomial $P\in\KK[x]$, a rational function $F=A/B\in\KK(x)$, a square-free factor $G$ of $B$ s.t. $\gcd(G,A+B'+iBG'/G )=1$ for all $i\in\ZZ$, a positive integer $k$.}{A polynomial $R\in\KK[x]$ such that $P=G^k(R+q'+Fq)$ for some $q\in\KK[x][G^{-1}]$.}
        \State $R\gets P$;
        \For{$i\gets 1$ to $k$}
          \State $C\gets A+(i-k-1)BG'/G +B'$;
          \State Write $R=QC+VG$ with $\deg Q<\deg G$;
          \State $R\gets (R-Q'B-QC)/G$;\Comment{$\frac{P}{G^k}\Phi=\frac{R}{G^{k-i}}\Phi+(q\Phi)'$}
          \EndFor
    \State \Return $R$.
\end{algoenv}
%  \caption{Partial Hermite-like reduction of a rational function\label{algo:BasicReduction}}
\end{algo}

In terms of integrals, our Algorithm~\textsf{HermiteReduction} lets one change $H^{n+1}$ into~$H^n$ in the integral,  writing
\[\int_\gamma{P(n,x)H(x)\Phi(n,x)\,dx}=\int_\gamma{\tilde{P}(n,x)\Phi(n,x)\,dx},\]
for some polynomial $\tilde{P}$.
It relies on a sequence of elementary steps (\textsf{BasicReduction}) based on the following lemma.
\begin{lem}\label{lemma:BasicReduction}Let $\Phi$, $A$ and $B$ be as in Eq.~\eqref{eqn:input}. Then for any $G$ dividing $B$ and satisfying $\gcd(G,A+B')=1$, there exist polynomials $Q$ and $R$ in $\KK[x]$ such that
\begin{equation}\label{eq:Hermite-Phi}
P\Phi=RG\Phi+(QB\Phi)'.
\end{equation}
\end{lem}
\begin{proof}
The hypothesis $\gcd(G,A+B')=1$ implies the existence of polynomials $U,Q\in\KK[x]$ such that $P=UG+Q(A+B')$.
Then the derivative $(QB\Phi)'$ expands as
\[
\frac{(QB\Phi)'}{\Phi}=Q'B+Q(A+B')=(Q'B/G-U)G+P,
\]
which has exactly the form of Equation~\eqref{eq:Hermite-Phi}.
\end{proof}
The crucial condition $\gcd(G,A+B')=1$ required to apply this lemma does not hold for arbitrary $A$ and $B$ and divisor $G$ of $B$, but when $G$ is square-free, it is a consequence of the presence of $n$ in Eq.~\eqref{eqn:input}, as shown in the following.
\begin{lem}\label{lemma:A+B'}
Let $\Phi$, $A$ and $B$ be as in Eq.~\eqref{eqn:input}. Then for any square-free polynomial $G$ in $\kk[x]$ dividing the denominator of~$H$, $\gcd(G, A+B')=1$.
This is also true if $A/B$ is replaced by the reduced form of $A/B+iG'/G$ for some $i\in\ZZ$.
\end{lem}
\begin{proof}
If $H'=0$ then the denominator of $H$ is~1 and then $G=1$ and the property holds.

Otherwise, let first $G$ be an irreducible factor of the denominator of $H$, so that there exist an integer $k$ and polynomials $H_1,H_2\in\kk[x]$ such that
\[\frac{A}{B}=nk\frac{G'}{G}+n\frac{H_1}{H_2}+\frac{S}{T},\]
with $\gcd(G,H_2)=1$ and $B=\operatorname{lcm}(G,H_2,T)$. Write $B=G^\nu\tilde{B}$ with $\gcd(G,\tilde{B})=1$. Then 
\[A=nkG'G^{\nu-1}\tilde{B}+nH_1G^\nu({\tilde{B}}/{H_2})+S({G^\nu\tilde{B}}/{T}).\]
Reducing $A+B'$ modulo $G$ then yields
\begin{equation}\label{eq:A+B'}
A+B'\equiv (nk+\nu)G'G^{\nu-1}\tilde{B}+S({G^\nu\tilde{B}}/{T}) \bmod G.
\end{equation}
Since $G$ does not depend on $n$, $G\cmid A+B'$ would imply that both $G\cmid G'G^{\nu-1}\tilde{B}$ and $G\cmid SG^\nu\tilde{B}/T$. The first condition implies $\nu>1$, which forces that $G^\nu\cmid T$, making $G\cmid S$ necessary too, a contradiction. 
This reasoning also proves the result when adding integer multiples of $G'/G$ to the fraction $A/B$, which adds an integer to $nk+\nu$ in Eq.~\eqref{eq:A+B'}.

%\begin{proof}
If $G$ is only assumed to be a square-free divisor of the denominator of $H$, then the property holds for each of its irreducible factors, and thus $A+B'$ is invertible modulo their product~$G$ by the Chinese remainder theorem.
\end{proof}

\begin{algo}
  Algorithm \textsf{HermiteReduction}($P$,$H$,$S/T$)
  
  \begin{algoenv}{A polynomial $P\in\kk(n)[x]$,\\ two rational functions $H$ and $S/T$ in $\kk(x)$.}{A polynomial $R\in\kk(n)[x]$ such that $P=R/H+nQH'/H+QS/T+Q'$ for some $Q\in\kk(n,x)$.}
    \State Compute the square-free decomposition of the denominator of $H$: $g=g_1g_2^2\dots g_m^m$;
    \State Compute the corresponding partial fraction decomposition: $H=U+\sum_{k=1}^m{U_k/g_k^k}$;
    \State $R\gets PU$;
    \State $K\gets S/T+(n-1){H'}/{H}$;
      \For{$k\gets 1$ to $m$}
      \State $r_k\gets\textsf{BasicReduction}(PU_k,K,g_k,k)$;
      \EndFor
    \State\Return $R+r_1+\dots+r_m$.
  \end{algoenv}
% \caption{Hermite reduction of a proper hyperexponential function\label{algo:fullreduc}}
\end{algo}

\smallskip\par\noindent \textbf{Correctness of} \textsf{BasicReduction} \textbf{and} \textsf{HermiteReduction.}
Algorithm~\textsf{HermiteReduction} treats each square-free factor of the denominator of $H$ separately, while
the second part of the lemma is used in Algorithm~\textsf{BasicReduction} to deal with multiplicities. If $G$ is a square-free factor of multiplicity $k$, then Lemma~\ref{lemma:BasicReduction} is used successively with $A/B$ the reduced form of the logarithmic derivatives of $\Phi/G^k,\Phi/G^{k-1},\dots,\Phi/G$, thus rewriting $P\Phi/G^k$ as $R\Phi$ up to a derivative. 

\subsection{Mixed Creative Telescoping}\label{sec:mixedCT}
Combining confinement and Hermite reduction gives the final result.

\begin{thm}\label{thm:mixedct-order} Let $P$ be a polynomial in $\kk[n,x]$ and ${\Phi}$, $A$ and $B$ as in Eq.~\eqref{eqn:input} and $\delta$ as in Eq.~\eqref{eq:eqdelta}. Then, $F_n=P(n,x){\Phi}$ admits a telescoper of order bounded by $\delta$.
\end{thm} 
\begin{proof}
By Lemma~\ref{lemma:BasicReduction}, Algorithm~\textsf{HermiteReduction} can be used to rewrite all the shifts $F_n,F_{n+1},F_{n+2},\dotsc$ under the form $R\Phi$ modulo derivatives, with $R\in\kk(n)[x]$. The necessary condition to apply the algorithm is satisfied at each step thanks to Lemma~\ref{lemma:A+B'}. By Lemma~\ref{lemma:confinement}, $F_n,F_{n+1},\dots,F_{n+\delta}$ are linearly dependent modulo derivatives. A linear relation between them provides a telescoper of order at most $\delta$.
\end{proof}

\begin{algo}
  Algorithm \textsf{MixedCT}($P,H,S/T$)
  
  \begin{algoenv}{A polynomial $P\in\kk[n,x]$,\\ two rational functions $H$ and $S/T$ in $\kk(x)$.}{
  A $r$-tuple $(c_0,\dots,c_{r-1})$ such that $P({n+r},x)H^{r}-\sum_{i=0}^{r-1}c_{i}P(n+i,x)H^{i}=nQH'/H+QS/T+Q'$ for some $Q\in\kk(n,x)$.}
    \State $F\gets S/T+(n-1)H'/H$;
    \State $R_0\gets\textsf{Confinement}(P\mid_{n\mapsto n-1},F)$;
    \For {$k\gets 0,\dots$}
      \If {$\operatorname{rank}_{\kk(n)}(R_0,R_1,\dots,R_k)<k+1$}
        \State Solve $\sum_{i=0}^{k-1}c_iR_i=R_k$ for $c_0,\dots,c_{k-1}$ in $\kk(n)$;
        \State \Return $(c_0,\dots,c_{k-1})$. % $S_n^s-\sum_k{c_k S_n^k}$
      \EndIf  
      \State $P\gets R_{k}\mid_{n\mapsto n+1}$;
      \State $P\gets\textsf{HermiteReduction}(P,H,S/T)$;
      \State $R_{k+1}\gets\textsf{Confinement}(P,F)$;
    \EndFor
  \end{algoenv}
%\caption{Mixed Creative Telescoping\label{algo:MixedCT}}
\end{algo}

Algorithm~\textsf{MixedCT} implements that proof. In practice, for efficiency purposes, the reduction of $F_{n+i}$ is obtained by applying \textsf{HermiteReduction} to the shift of the confined reduction of $F_{n+i-1}$.

\begin{thm}
	With the notation of Theorem~\ref{thm:mixedct-order}, further assume that $\Phi'/\Phi$ does not have any positive integer residue.
	Then Algorithm~\textsf{\emph{MixedCT}}($P$,$H$,$S/T$) computes a minimal-order telescoper for $P\Phi$.
\end{thm}

\begin{proof}
	Consider the minimal-order monic telescoper 
	\[L(n,S_n)=S_n^r-\sum_{i=0}^{r-1}c_i(n)S_n^i\]
	of $P\Phi$ and its certificate $C\in\kk(n,x)$ such that \[L(n,S_n)\cdot(P\Phi)=(C\Phi)'.\]
By Lemma~\ref{lemma:confinement}, there exist $R$ and $K$ satisfying \[L(n,S_n)\cdot(P\Phi)=R\Phi+(K\Phi)',\] where $K\in\kk(n,x)$ %is a rational function 
and
	\[R=R_r-\sum_{i=0}^{r-1}{c_i(n)R_i}\]
	is a linear combination of the $R_i$'s of Algorithm \textsf{MixedCT}. It follows that $R\Phi=((C-K)\Phi)'$ and Proposition~\ref{lem:confinement soundness property} then implies that $R=0$. Thus, this linear combination is detected by the algorithm, producing the output $(c_0,\dots,c_{r-1})$.
\end{proof}

\paragraph{Certificates} Algorithm~\textsf{MixedCT} as given here computes the certificate, although not in a normalized form. We chose to only output the telescoper, but it would be possible to return the certificate as well (and normalize it or not).

\section{Degree Bounds}\label{sec:degreebounds}

We now review more precisely the algorithms and obtain bounds on the degrees at each step. The first part of this section consists of technical results that are needed for the complexity analysis in the next section. Then, at the end of section~\ref{sec:mixedct-bound-deg}, we give a bound on the degree of the telescoper produced by Algorithm~\textsf{MixedCT}.

\subsection{Confinement}
\begin{lem}\label{lem:bound-Confinement}
Let $A$ and $B$ be as in Eq.~\eqref{eqn:input} and let $P$ be a polynomial in $\kk(n)[x]$. Let $R$ be the polynomial returned by Algorithm~\textsf{\emph{Confinement}}. Then $\Rdeg R-\Rdeg P$ is at most
\[\begin{cases}(\deg_xP-\deg_xA+1)\frac{[1]}{[1]},\quad&\text{if $\deg_x B\le\deg_xA+1$,}\\
{\left[\left\lfloor\frac{\deg_xP-\deg_xB+1}{\deg_xB-\deg_xA-1}\right\rfloor+1\right]}/{[0]}&\text{otherwise.}
\end{cases}\]
\end{lem}
\begin{proof} 
The proof is a case by case analysis of Algorithm~\textsf{Confinement}; we use its notation.

When $\delta = \deg_xA$, the recurrence for $q_i$ has a summand $a_{\delta-1}q_{i+1}$ except when $i=d$, while $c$ has $a_\delta$ for summand. Thus by induction, the degree in $n$ of the numerator and denominator of $q_{d-i}$ increase by $1$ at each step. Since there are $d+1$ steps, $\Rdeg Q-\Rdeg P$ is bounded by $[d]/[d+1]$. The result for $R$ then follows from Equation~\eqref{eq:confinement2}.

When $\delta = \deg_xB-1$, the coefficients $a_{\delta-j}$ are zero for $j< \delta-\deg_xA$ and $c$ has degree $0$ in $n$. By induction on $i$, $q_{d-i}$ has degree that changes (by increases of $1$) only when $i\equiv 0 \mod{\delta-\deg_xA}$. More explicitly, $\Rdeg q_{d-i}-\Rdeg P$ is bounded by $[\lfloor i/(\delta-\deg_xA)\rfloor]/[0]$. Again, the conclusion for~$R$ follows from Equation~\eqref{eq:confinement2}.
\end{proof}

\subsection{Hermite Reduction}

In order to track the degrees in $n$ of the polynomials involved in Algorithms~\textsf{BasicReduction} and~\textsf{HermiteReduction}, we need to look deeper into the modular inversions involved. This is done in the next lemma, using the same notation as in the discussion preceding Lemma~\ref{lemma:A+B'}.
\begin{lem}\label{lem:modular inverse}
Let $A$, $B$ and $H$ be as in Eq.~\eqref{eqn:input}. Let $G$ be an irreducible factor of the denominator of $H$, and $\nu$ the $G$-adic valuation of $B$. Let $P$ be in $\kk(n)[x]$ and let $Q$ be the polynomial such that $\deg_xQ<\deg_x(G)$ and $P\equiv Q(A+B')\bmod G$. Then $\Rdeg Q-\Rdeg P$ is bounded by
%With the preceding notation and assumptions, decompose $Q^\star$ as %$Q=Q_1Q_2Q_3$, where
%\[Q_1=\gcd(Q^\star,T,T'),\ Q_2=\gcd(Q^\star/Q_1,T),\  Q_3=Q^\star/(Q_1Q_2).\
\[\begin{cases}[0]/[0],\qquad&\text{if $\nu>1$};\\
[0]/[1],\qquad&\text{if $\nu=1$ and $G\ndiv T$};\\
[\deg_x(G)-1]/[\deg_x(G)],\quad&\text{if $\nu=1$ and $G\cmid T$.}
\end{cases}
\]
\end{lem}
\begin{proof}
The existence of $Q$ follows from Lemma~\ref{lemma:A+B'}. Notice first that $\Rdeg Q-\Rdeg P=\Rdeg (A+B')^{-1}$ (where $(A+B)^{-1}$ denotes an inverse mod $G$). Indeed, writing $p(n)P= P_0(x) + \cdots + P_d(x) n^d$ with $p(n)$ the denominator of $P$, we see that $p(n)Q = (A+B')^{-1}(P_0\bmod C) + \cdots + (A+B')^{-1} n^d (P_d\bmod C)$ has rational degree in $n$ at most $[d]/[0]+\Rdeg(A+B')^{-1}$.

Thus we just need to bound $\Rdeg(A+B')^{-1}$. To do so, we take a closer look at Equation~\eqref{eq:A+B'}.
If $\nu>1$, the equation becomes
\[A+B'\equiv S\frac{\tilde{B}}{T/G^\nu} \bmod G,\]
so that $(A+B')^{-1}$ does not depend on $n$ in this case.

If $\nu=1$ and $G\ndiv T$, the equation becomes
\[A+B'\equiv (nk+1)G'\tilde{B}.\]
Then, $(A+B')^{-1}\equiv (G'\tilde{B})^{-1}/(nk+\nu)$ has rational degree $[0]/[1]$.
Finally, if $\nu=1$ and $G\cmid T$, the result follows from Lemma~\ref{bounds-factors} below.
\end{proof}

\begin{lem}\label{bounds-factors}
Let $A,B,C$ be polynomials in $\kk[x]$ such that $C$ is relatively prime with at least one of $A$ or $B$. 
Let $U,V \in \kk(n)[x]$  be such that 
\begin{equation}\label{eq:bezout}
1=U(A n+B)+VC, \quad \text{with} \quad \deg_x U<\deg C.
\end{equation}
Then $\Rdeg U$ is bounded by $[\deg C-1]/[\deg C]$.
\end{lem}

\begin{proof}
Let $R(n) \neq 0$ denote the resultant with respect to $x$ of $C$ and $A n + B$.
Then $R = S(An+B) + TC$ for some $S,T$ in $\kk[n,x]$ with $\deg_x S < \deg C$.
Moreover, by Cramer's rule applied to the Sylvester matrix of $C$ and $A n + B$, we have $\deg R \leq \deg C$ and $\deg_n S < \deg C$.

Now denote by~$q(n)$ the monic  denominator of $U$, and by $\tilde{U} \in \kk[n,x]$ its numerator. We need to prove that $\deg q \leq \deg C$ and $\deg_n \tilde{U} < \deg C$. Equalities
\[1 = U(A n+B)+VC, \quad  1 = S/R(A n+B)+ (T/R) C\] imply by subtraction that $C$ divides
$(An+B)(U-S/R)$ in $\kk(n)[x]$. Since $C$ is coprime with $An+B$, this implies that $C$ divides $U-S/R$. As $\deg C < \deg_x(U-S/R)$, this shows that $U = S/R$. In particular, $q$ divides $R$. It follows that $\deg q \leq \deg R \leq \deg C$ and $\deg_n \tilde{U} = \deg_y (qS/R) \leq \deg_n S < \deg C$, which concludes the proof.
\end{proof}

\begin{lem}\label{lem:bound-BasicReduction}
With the same notation as in Lemma~\ref{lem:modular inverse}, and assuming that $G$ has multiplicity $k$ in the denominator of $H$, the output $R$ of \textsf{\emph{BasicReduction}}$(P,A/B,G,k)$ satisfies
\[\deg_xR \le \max(\deg_xP-k\deg_xG,\deg_xA-1,\deg_xB-2),\]
and $\Rdeg R-\Rdeg P$ is bounded by
\[\begin{cases}k[1]/[0],\qquad&\text{if $\nu>1$};\\
k[1]/[1],\qquad&\text{if $\nu=1$ and $G\ndiv T$};\\
k\deg_x(G)[1]/[1],\quad&\text{if $\nu=1$ and $G\cmid T$.}
\end{cases}
\]
\end{lem}
\begin{proof}
	Both bounds follow from Lemma~\ref{lem:modular inverse}. For the degree in $n$, the polynomial $Q$ of the algorithm is obtained from a modular inverse as above, that is multiplied by $R$, and then by $C$ that has rational degree $[1]/[0]$ in $n$. The bound directly follows in the first and third cases since the condition on $\nu$ and $G$ is preserved at each step. In the second case, the condition $G\nmid T$ is not preserved, but writing $J=iG'/G+S/T$ at each step shows that the degree in $n$ still increases by $[1]/[1]$ only. For the degree in $x$, the bound is obtained by bounding each term in the expression of $R$ that is used in the algorithm.
\end{proof}

\begin{lem}\label{lem:bound-HermiteReduction}
With  the same notation, write the denominator $g$ of $H$ as $g=efh$, where
\[e=\gcd(g,T,T'),\quad f=\gcd(g/e,T).\]
Also denote $m$ the highest multiplicity of the roots of $g$, and %$X=X_1X_2^2\cdots X_m^m$ the square-free decomposition of $X$ for $X=e,h$.
let $e=e_1e_2^2\cdots e_m^m$ and $h=h_1h_2^2\cdots h_m^m$ be the square-free decompositions of $e$ and $h$.
% $X$ for $X=e,h$.

Then, the result of Algorithm \textsf{\emph{HermiteReduction}}$(P,H,S/T)$ is a polynomial in $x$ of degree at most 
\[\max(\deg_xP+\deg_x^\infty H,\deg_xP-1,\deg_xA-1,\deg_xB-2).\]
Seen as a rational function in~$n$, it has degree at most
\[\Rdeg P + \frac{[\max_{e_k\neq1}{k}+\sum_{h_k\neq 1}{k}+\deg_xf]}{[\sum_{h_k\neq 1}{k}+\deg_xf]}.\]
Moreover, the certificate $Q$ in the algorithm satisfies $Q=qB/(gH)$ for some polynomial $q$ such that $\deg_xq<\deg_xg$.
\end{lem}

\begin{proof}
Following the notation of Algorithm~\textsf{HermiteReduction}, the partial fraction decomposition of $H$ produces $U$ with $\deg_xU\le \deg_x^\infty H$ and $U_k$ with $\deg_xU_k<\deg_xg_k$. $PU$ obviously satisfies the bounds. Now write the square-free decomposition $f=f_1f_2^2\cdots f_m^m$. By Lemma~\ref{lem:bound-BasicReduction}, $\Rdeg r_k \le \Rdeg P+[k(\mathbf{1}_{e_k\neq1}+\mathbf{1}_{h_k\neq 1}+\deg_xf_k)]/[k(\mathbf{1}_{h_k\neq 1}+\deg_xf_k)]$. Normalizing $R+r_1+r_2+\dots+r_m$ then yields the result. The bound for the degree in $x$ is obtained by bounding separately each term using Lemma~\ref{lem:bound-BasicReduction}. The form of $Q$ follows from the fact that the certificate of the $k$-th call to \textsf{BasicReduction} has the form $q_kB/(g_k^kH)$ with $\deg_xq_k<\deg_x(g_k^k)$.
\end{proof}

\subsection{Mixed Creative Telescoping}\label{sec:mixedct-bound-deg}

\begin{lem}\label{lem:bound-matrix}
	With the notation of Eq.~\eqref{eqn:input}, Eq.~\eqref{eq:eqdelta}, Algorithm \emph{\textsf{MixedCT}}, and $d_H=\deg_x^\infty H$, for all $i$ in $\left\{1,2,\dots,\delta \right\}$ we have	
	\[\Rdeg R_i\le \deg_nP\cdot\frac{[1]}{[0]}+\alpha+i(\beta+\gamma)\]
	where
	\[\alpha=\begin{cases}\max(\deg_xP-\delta+1,0)\cdot\frac{[1]}{[1]}\quad&\text{if $\delta=\deg_xA$,}\\
	{\left[\max\left(\left\lfloor\frac{\deg_xP-\delta}{\delta-\deg_xA}\right\rfloor+1,0\right)\right]}/{[0]}&\text{otherwise.}
	\end{cases}\]
	\[\beta= \frac{[\max_{e_k\neq1}{k}+\sum_{h_k\neq 1}{k}+\deg_xf]}{[\sum_{h_k\neq 1}{k}+\deg_xf]}\]
	\[\gamma=\begin{cases}(d_H+1)\cdot\frac{[1]}{[1]}\quad&\text{if $d_H\ge 0$ and $\delta=\deg_xA$,}\\
	{\left[\left\lfloor\frac{d_H}{\delta-\deg_xA}\right\rfloor+1\right]}/{[0]}&\text{if $d_H\ge 0$ and $\delta=\deg_xB-1$,}\\
	0\quad& \text{otherwise}.
	\end{cases}\]
\end{lem}

\begin{proof}
	By Lemma~\ref{lem:bound-Confinement}, the initial confinement increases $\Rdeg P$ by $\alpha$.
	\textsf{HermiteReduction} is then always used with an input polynomial of degree less than $\delta$. By Lemma~\ref{lem:bound-HermiteReduction} each call to \textsf{HermiteReduction} increases $\Rdeg P$ by at most $\beta$, and produces an output of degree at most $\delta$ if $d_H<0$ or $\delta+d_H$ if $d_H\ge 0$. Thus the confinement is only necessary in the latter case. Plugging this bound into Lemma~\ref{lem:bound-Confinement} shows that each call to \textsf{Confinement} increases $\Rdeg P$ by at most~$\gamma$.
\end{proof}

\subsection{Degree bound on the telescoper}

\begin{lem}\label{lem:degree bound of certificate} 
	With the notation of~\eqref{eqn:input}, let $(c_0,\dots,c_{r-1})$ be the output of \emph{\textsf{MixedCT}}($P$,$H$,$S/T$) and write $H=f/g$ with $\gcd(f,g)=1$.
	
	Then there exists a polynomial $Q\in\kk(n)[x]$ such that
	\[\left(S_n^r-\sum_{i=0}^{r-1}{c_iS_n^i}\right)(P\Phi)=\left(\frac{Q}{g^r}B\Phi\right)',\]
	with $\deg_x Q\le r\deg_xH +\max(\deg_xP-\delta,0)-1$.
\end{lem}

\begin{proof}
	Tracking the certificates of the various rewritings in \textsf{MixedCT}, it suffices to show that for all $i\in\left\{0,\dots,r\right\}$
	\[P(n+i-1,x)H^{i}\frac{\Phi}{H}=R_i\frac{\Phi}{H}+\left(\frac{Q_i}{g^i}B\frac{\Phi}{H}\right)'\]
	for some polynomial $Q_i$ such that 
	\[\deg_x Q_i\le\max(\deg_xP-\delta,0)+i\deg_xH-1.\]
	
	This is obvious for $i=0$ (initial confinement). Assume this is true for $i-1$, then the next Hermite reduction writes
	\[R_{i-1}(n+1,x)\Phi=\tilde{R}_i\frac{\Phi}{H}+\left(\frac{\tilde{Q}_i}{g}B\frac{\Phi}{H}\right)'\]
	with $\deg_x\tilde{Q}_i<\deg_xg$ and $\deg_x\tilde{R}_i\le \delta-1+\delta_H$, where $\delta_H=\max(\deg_x^\infty H,0)$ (see Lemma~\ref{lem:bound-HermiteReduction}).
	As for the confinement,
	\[\tilde{R}_i\frac{\Phi}{H}=R_i\frac{\Phi}{H}+\left(\overline{Q}_iB\frac{\Phi}{H}\right)',\]
	with $\deg_x\overline{Q}_i\le\delta_H-1$ by Lemma~\ref{lemma:confinement}.
	Thus, the property is satisfied for $i$ with
	\[Q_i=fQ_{i-1}(n+1,x)+\tilde{Q}_ig^{i-1}+\overline{Q}_ig^i,\]
	from which follows the bound on the degree of $Q_i$.
\end{proof}

\begin{thm}\label{thm:bound-degree} With the notation of~\eqref{eqn:input}, the telescoper $L$ produced by Algorithm \emph{\textsf{MixedCT}}($P$,$H$,$S/T$) satisfies
	\[\deg_nL\le r(\deg_nP+\deg_xH)+\max(\deg_xP-\delta,0),\]
	where $r$ is the order of $L$.
\end{thm}

\begin{proof}
	Write $H=f/g$ with $\gcd(f,g)=1$. Rewriting Lemma~\ref{lem:degree bound of certificate} in terms of polynomials yields
	\[P(n+r)f^r-\sum_{i=0}^{r-1}{c_iP(n+i)f^ig^{r-i}}=(QB)'-rQ\frac{Bg'}{g}+QA\]
	for some $Q=\sum_{i=0}^s{q_ix^i}$ with $s=\max(\deg_xP-\delta,0)+r\deg_xH-1$.
	This equation is a linear system with two blocks of unknowns: $c_0,\dots,c_{r-1}$ with coefficients of degree bounded by $\deg_n P$ and $q_0,\dots,q_s$ with coefficients of degree $1$, which by Hadamard's bound yields $\deg_nc_i\le r\deg_nP+s+1$, whence the theorem.
\end{proof}

\section{Complexity}\label{sec:complexity}

We will rely on some classical complexity results for the basic operations on polynomials and rational functions. Standard references for these questions are the books~\cite{GaGe03} and~\cite{BuClSh97}.
We will also use the fact that linear systems with polynomial coefficients can be solved efficiently using Storjohann and Villard's algorithm~\cite{StVi05}.
The needed results are summarized in the following lemma.

\begin{lem}\label{lem:complexity}
Addition, product and differentiation of rational functions in $\KK(x)$ of regular degree less than d,
as well as extended gcd and square-free decomposition in $\KK[x]_d$ can be performed using $\softO(d)$ operations in~$\KK$.

The kernel of a $s\times(s+1)$ matrix with polynomial entries in $\kk[x]_d$ can be solved using $\softO(s^\omega d)$ operations in $\kk$.
\end{lem}

\subsection{Confinement}

\begin{lem}\label{lem:complexity-confinement} With $P\in\KK[x]$, $A/B\in\KK(x)$ as input, and $\delta=\max(\deg_xA,\deg_xB-1)$,
\emph{\textsf{Confinement}} performs at most
\[\bigO(\delta\deg_xP)\] operations in $\KK$.
\end{lem}
\begin{proof}
Each iteration of the loop performs $\bigO(\delta)$ operations in $\KK$ and the loop is executed $\deg_xP-\delta+1$ times. The computation of $R$ then needs $\softO(\deg_x P)$ operations in~$\KK$.
\end{proof}

\subsection{Hermite Reduction}
\begin{lem} With the same notation as in the preceding lemma, $G$ a square-free factor of $B$ and $k$ a positive integer, \textsf{\emph{BasicReduction}}($P$, $A/B$, $G$, $k$) performs at most
\[\softO(k(\deg_xP+\delta))\]
operations in $\KK$.
\end{lem}

\begin{proof}
The costly steps are the gcd computations, which according to Lemma \ref{lem:complexity} can be performed using $\softO(\deg_xP+\delta)$ operations in $\KK$. The result then follows since there are $k$ gcd computations.
\end{proof}

\begin{lem} With the same notation as in Lemma~\ref{lem:bound-HermiteReduction}, set $\epsilon=\sum_{g_k\neq1}{k}$. Then \textsf{\emph{HermiteReduction}}($P$,$H$,$S/T$) performs at most
	\[\softO(\max(\epsilon\deg_xP,d_H)+\deg_x g+\epsilon\delta)\] operations in $\KK=\kk(n)$.
\end{lem}
\begin{proof}
By Lemma \ref{lem:complexity}, the square-free decomposition of~$g$ can be computed in $\softO(\deg_xg)$ operations and the product $PU$ is computed at a cost $\softO(\max(\deg_xP,\deg_x^\infty H))$. By the preceding lemma, the $k$-th call to \textsf{BasicReduction} uses $\softO(k(\deg_xP+\deg_xg_k+\delta))$ operations. The announced bound is then obtained by summation.
\end{proof}

\subsection{Mixed creative telescoping}

For the sake of simplicity, Algorithm~\textsf{MixedCT} searches for telescopers for all the possible orders, starting from $0$. In practice, a more efficient variant consists in carrying a dichotomic search of the order between $0$ and $\delta$. This way the complexity is that of the last step up to a logarithmic factor. Here, we analyze this variant.

\begin{lem} With the same notation as in Section~\ref{sec:mixedct-bound-deg}, and $\epsilon=\sum_{g_k\neq1}{k}$, the number of operations in $\KK=\kk(n)$ performed by \textsf{\emph{MixedCT}}($P$,$H$,$S/T$) is
	\begin{compactenum}
		\item $\softO(\delta\max(\epsilon\deg_xP,d_H)+\delta\deg_xg+\epsilon\delta^2+\delta^\omega)$ if $d_H<0$;
		\item $\softO(\delta\max(\epsilon\deg_xP,\delta d_H)+\delta\deg_xg+\epsilon\delta^2+\delta^3)$ if $d_H\ge 0$.
	\end{compactenum}
\end{lem}

\begin{proof}
	When $d_H<0$, by Lemma~\ref{lem:bound-HermiteReduction} there are no confinement steps and the construction of the system to solve amounts to $\delta$ Hermite reductions. The algorithm then computes a vector in the kernel of a $\delta\times(\delta+1)$ matrix, which by Lemma~\ref{lem:complexity} can be performed in $\softO(\delta^\omega)$ operations in $\KK$, hence the complexity.
	
	When $d_H\ge 0$, we have to add the cost of the confinement steps, which by Lemma~\ref{lem:bound-HermiteReduction} are performed on polynomials of degree at most $\delta+d_H$. There are at most $\delta+1$ calls to confinement, so the result follows from Lemma~\ref{lem:complexity-confinement}.
\end{proof}

\begin{thm}\label{thm:global-complexity}
	With the same notation, set $\mu=\max(a,b)$ where $\deg_n P\cdot[1]/[0]+\alpha+\delta(\beta+\gamma)=[a]/[b]$, the number of operations in $\kk$ performed by Algorithm \textsf{\emph{MixedCT}} is
	\[\softO(\mu(\epsilon\delta\deg_xP+\delta\deg_xg+\epsilon\delta^2+\delta^\omega)),\]
	if $d_H<0$, or
	\[\softO(\mu(\epsilon\delta\deg_xP+\delta\deg_xg+\epsilon\delta^2+\delta^3+\delta^2d_H))\]
	if $d_H\ge 0$.
\end{thm}
\begin{proof}
The result follows directly from the preceding lemma and the fact that all the elements of $\KK=\kk(n)$ appearing in the construction of the linear system have numerator and denominator of degree in $n$ bounded by $\mu$. The cost of solving is then $\softO(\delta^\omega\mu)$ operations in $\kk$ by Lemma~\ref{lem:complexity}.
\end{proof}

The complexity result from Theorem~\ref{thm:main theorem} follows directly.

\section{Experiments and Applications}\label{sec:applications}
\subsection{Various Integrals} % (fold)
\label{sub:various_integrals}
\begin{example}The Jacobi polynomials have the following integral representation, up to a factor that does not depend on~$n$:
\[\oint{\left(\frac{z^2-1}{2(z-x)}\right)^n(1-z)^\alpha(1+z)^\beta\frac{dz}{z-x}} ,\]
with a contour enclosing $z=x$ once in the positive sense~\cite[18.10.8]{OlverLozierBoisvertClark2010}. It is well-known that the Jacobi polynomials satisfy a recurrence of order~2. Theorem~\ref{thm:mixedct-order} is tight in that case: it predicts a bound~2 on the order of the telescoper, since the logarithmic derivative of the integrand has numerator of degree~2 and denominator of degree~3.
\end{example}
Note however that in such an example, a more direct and efficient way to obtain a recurrence is to change the variable $z$ into $x+u$ making the integral that of the extraction of the $n$-th coefficient in a hyperexponential term. That is achieved easily by translating the first order linear differential equation it satisfies into a linear recurrence. The only difference is that our method guarantees the minimality of the telescoper.
\begin{example}
A typical example where several of the difficulties are met at once is the term
\[\left(1+\tfrac{x}{n^2+1}\right)\left(\tfrac{(x+1)^2}{(x-4)(x-3)^2(x^2-5)^3}\right)^{\!n}\!\sqrt{x^2-5}\,e^{\frac{x^3+1}{x(x-3)(x-4)^2}}.\]
Our code\footnote{The code and Maple worksheet for these examples can be found at \url{http://mixedct.gforge.inria.fr}.} finds a telescoper of order~9 and degree~90 in 1.4 sec. The only other code we are aware of that can perform this computation is Koutschan's \texttt{HolonomicFunctions} package\cite{Koutschan14}, which takes more than 3 min. This is by no means a criticism of this excellent and very versatile package but rather to indicate the advantage of implementing specialized algorithms like ours for this class.
\end{example}
% subsection various_integrals (end)

\subsection{Inversion of rational functions}
Here is a generalization of Manivel's lemma~\cite{Furter15}, which led us to this work.
It gives an efficient way to compute the recurrence satisfied by the coefficients of the compositional inverse of a rational function. The starting point is Lagrange inversion. Let $f\in\QQ(x)$ be a rational function such that $f(0)=0$, and denote by $f^{(-1)}$ its compositional inverse. By Cauchy's formula, the $n$-th coefficient $u_n$ of $f^{(-1)}$ is given by
\[u_n=\frac{1}{2\pi i}\oint{f^{(-1)}(x)\frac{\mathrm{d}x}{x^{n+1}}},\]
where the contour is a small circle around the origin. Integrating by parts and then using the change of variables $x=f(u)$ yields
\[u_n=\frac{1}{2\pi in}\oint{\frac{f^{(-1)'}(x)\mathrm{d}x}{x^n}}=\frac{1}{2\pi in}\oint{\frac{\mathrm{d}u}{f(u)^n}}.\]
Thus a recurrence for $u_n$ can be computed with Algorithm \textsf{MixedCT}. Theorem~\ref{thm:mixedct-order} and Theorem~\ref{thm:bound-degree} then provide bounds for the order and degree of this recurrence.

\begin{thm}\label{th:invrat}
	Let $f\in\QQ(x)$ be a rational function such that $f(0)=0$. Write $f=P/Q$ and denote $P=P_1P_2^2\cdots P_m^m$ the square-free decomposition of $P$. Also denote $p,p^\star,q,q^\star$ the degrees of $P,P^\star,Q,Q^\star$ respectively.
	Then the Taylor coefficients of $f^{(-1)}$ satisfy a recurrence of order at most
	\[q^\star+p^\star-1\]
	and of rational degree in $n$ at most
	\[\frac{(q^\star+p^\star)(q^\star+p^\star+1)}{2}(\max(q-p,0)+\sum_{P_k\neq 1}{k})\cdot\frac{[1]}{[1]}.\]
\end{thm}

\begin{table}[h]
\begin{center}\begin{tabular}{rrrrr}
  \hline $k$& order & degree & coeffs & time \\[1mm]
  \hline\\[-3mm]
%1&    2&   3&    28& 0.04  \\ 
%2&    4&  10&   227& 0.13  \\ 
%3&    6&  22&   424& 0.34  \\ 
%4&    8&  39&  1034& 0.87  \\ 
5&   10&  61&  1759& 2.49  \\ 
6&   12&  88&  2440& 4.64  \\ 
7&   14& 120&  3778& 13.36 \\
8&   16& 157&  4666& 33.89 \\
9&   18& 199&  6192& 88.34 \\
10&  20& 246&  8364& 260.59\\
11&  22& 298& 10146& 628.21\\
12&  24& 355& 11802&1451.54\\
  \hline 
\end{tabular}\end{center}
\caption{Order, degree, bit size and timings for Example~\ref{ex:experiments1}\label{tab:exp1}}
\end{table}

\begin{example}\label{ex:experiments1}Experimental results on the family of rational functions 
$f_k=xP_k(x)^2/Q_k(x)$
with $P_k$ and $Q_k$ two dense polynomials of degree~$k$ and integer coefficients of absolute value bounded by 100 are presented in Table~\ref{tab:exp1}. The first column gives the index~$k$. The second one is the order of the minimal-order telescopers, which is as predicted by Theorem~\ref{th:invrat}. The next one gives the degree of the telescoper; it displays a quadratic growth, as predicted by Theorem~\ref{th:invrat}. The column ``coeffs'' gives the bit size of the largest coefficient of the telescoper, whose growth seems slightly more than quadratic. Finally, the time  (in seconds) taken by our implementation is given in the last column.
\end{example}

%\subsection{Non-minimal telescopers}

\smallskip

{\bf Acknowledgements.} We are grateful to the referees for their thorough work and helpful comments. This work has been supported in part by FastRelax ANR-14-CE25-0018-01.

%\balancecolumns

%\bibliographystyle{abbrv}
%\begin{small}
%\bibliography{BoDuSa}
%\end{small}

% \begingroup
% 
\let\oldbibliography\thebibliography
\renewcommand{\thebibliography}[1]{%
  \oldbibliography{#1}%
  \setlength{\itemsep}{1.5pt}%
}
\makeatletter
\def\section{%
    \@startsection{section}{1}{\z@}{-10\p@ \@plus -4\p@ \@minus -2\p@}% GM
    {14\p@}{\baselineskip 14pt\secfnt\@ucheadtrue}%
}
\makeatother

\small
\smallskip

%\bibliographystyle{abbrv}
%\bibliography{BoDuSa}

% \endgroup

\end{document}